\documentclass[10pt, doublecolumn]{IEEEtran}
\usepackage{epsfig,latexsym}
\usepackage{float}
\usepackage{indentfirst}
\usepackage{amsmath}
\usepackage{bm}
\usepackage{amssymb}
\usepackage{times}
\usepackage{subfigure}
\usepackage{psfrag}
\usepackage{hyperref}
\usepackage{cite}
\usepackage{lastpage}
\usepackage{fancyhdr}
\usepackage{color}
 \usepackage{amsthm}
\usepackage{bigints}
\newtheorem{Lemma}{Lemma}
\newtheorem{lemma}[Lemma]{$\mathbf{Lemma}$}

\newcounter{problem}
\newcounter{save@equation}
\newcounter{save@problem}
\makeatletter
\newenvironment{problem}
{\setcounter{problem}{\value{save@problem}}%
  \setcounter{save@equation}{\value{equation}}%
  \let\c@equation\c@problem
  \subequations
}
{\endsubequations
  \setcounter{save@problem}{\value{equation}}%
  \setcounter{equation}{\value{save@equation}}%
}

\begin{document}
\title{ { On the Application of BAC-NOMA to 6G  umMTC }}

\author{ Zhiguo Ding, \IEEEmembership{Fellow, IEEE} and  H. Vincent Poor, \IEEEmembership{Life Fellow, IEEE}\thanks{
  Z. Ding and H. V. Poor are  with the Department of
Electrical Engineering, Princeton University, Princeton, NJ 08544,
USA. Z. Ding
 is also  with the School of
Electrical and Electronic Engineering, the University of Manchester, Manchester, UK (email: \href{mailto:zhiguo.ding@manchester.ac.uk}{zhiguo.ding@manchester.ac.uk}, \href{mailto:poor@princeton.edu}{poor@princeton.edu}).

}\vspace{-1.5em}} \maketitle
\begin{abstract}
 This letter studies the application of backscatter communications (BackCom) assisted non-orthogonal multiple access (BAC-NOMA) to the envisioned sixth-generation (6G)   ultra-massive machine type communications (umMTC).  In particular, the proposed BAC-NOMA transmission scheme can realize  simultaneous energy and spectrum cooperation between uplink and downlink users, which is important to support massive connectivity and stringent energy constraints in umMTC. Furthermore,  a resource allocation   problem for maximizing the uplink throughput and suppressing the interference between downlink and uplink transmission is formulated as an optimization problem and the corresponding optimal resource allocation policy is obtained. Computer simulations are provided to   demonstrate the superior performance   of  BAC-NOMA. 
\end{abstract}\vspace{-1.5em}

\section{Introduction}
One of the key  communication scenarios to be supported by the envisioned sixth-generation (6G) mobile network is ultra-massive machine type communications (umMTC) \cite{you6g}.  The main feature of umMTC is that there are an extremely large number of low-power devices to be connected, e.g., there are expected to be  more than $10^7$ devices/km$^2$ and many of these devices are energy-constrained Internet of Things (IoT) sensors. How to serve this huge number of IoT devices is challenging, due to   spectrum scarcity  and energy constraints.   

To tackle the spectrum constraint, non-orthogonal multiple access (NOMA) and full-duplex (FD) transmission have  been proposed to encourage   spectrum cooperation among wireless  users \cite{mojobabook,7676258,7105651,8306094}. To tackle the energy constraint, various energy cooperation schemes have been proposed. For example, the use of simultaneous wireless information and power transfer (SWIPT) can ensure that  the transmission of an energy-constrained device  is powered by the energy harvested from  the  signals sent by a non-energy-constrained device \cite{yuanweijsac,Zhangruipower2013}. Alternatively,   backscatter communication (BackCom)   can also be used to realize energy cooperation among the devices, where one device's signal is used to excite  the BackCom circuit of another device   \cite{backnoma,8907447,8636518}. Note that compared to SWIPT,    BackCom  technologies are more  mature and have already been extensively applied to various IoT applications.  

This letter considers the combination of NOMA and BackCom, termed BAC-NOMA, and investigates its application to 6G umMTC. Unlike the existing schemes in \cite{backnoma,8907447,8636518}, an application of BAC-NOMA to FD uplink and downlink transmission is focused on in this letter.  In particular, from the spectrum cooperation perspective, FD is use to ensure spectrum sharing between uplink and downlink users. Unlike FD assisted    orthogonal multiple access (FD-OMA) which is to pair a single uplink device to a downlink user, the use of NOMA ensures that   multiple uplink   devices are served simultaneously.     From the energy cooperation perspective, the signals sent by the base station to  the downlink user are used to excite the BackCom circuits of  the uplink devices. As a result,  battery-less uplink transmission can be supported, and there is no need for the base station to send dedicated single-tone sinusoidal continuous waves, as in conventional BacCom. Note that the signals reflected by the BackCom uplink devices can cause interference to   downlink transmission, which motivates the resource allocation problem addressed  in the letter. In particular, an uplink throughput maximization problem is formulated, by using  the downlink user's   quality of service (QoS) requirement as a constraint.  The formulated optimization problem is not concave, but its optimal solution can be obtained   by recasting it to an equivalent     linear programming (LP) problem. Computer simulations are also provided to demonstrate the superior performance gain of BAC-NOMA over conventional schemes.  

 \begin{figure}[t]\centering \vspace{-0em}
    \epsfig{file=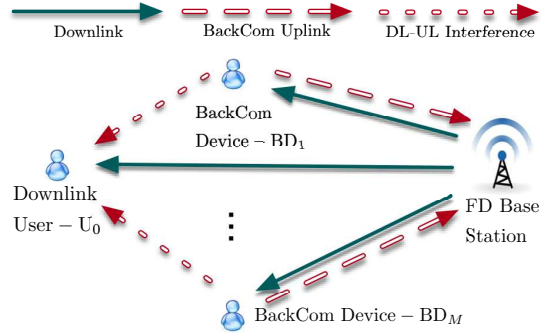, width=0.4\textwidth, clip=}\vspace{-0.5em}
\caption{ Illustration of the considered BAC-NOMA communication scenario.       }\label{fig1}\vspace{-1.5em}
\end{figure}

\vspace{-1em}
\section{System Model}
Consider a communication scenario with one full-duplex base station, one downlink user, denoted by ${\rm U_0}$, and $M$ uplink BackCom devices, denoted by ${\rm BD}_m$, $1\leq m \leq M$, as shown in Fig. \ref{fig1}. Each node is equipped with a single antenna.  For conventional BackCom, the base station needs to transmit single-tone sinusoidal continuous waves which are to excite the  circuit of a BackCom device. For   BAC-NOMA, the signals sent to the downlink user will be used to excite the circuits of the uplink BackCom devices. 
The signal received at    ${\rm BD}_m$ is given by $ \sqrt{P_0}h_ms_0$, and the signal reflected by ${\rm BD}_m$ is given by $\sqrt{\eta_m P_0}h_ms_0s_m$,
where $h_m$ denotes the channel gain between   ${\rm BD}_m$ and the base station, $P_0$ is the downlink transmit power, $s_0$ denotes ${\rm U}_0$'s signal, $s_m$ denotes ${\rm BD}_m$'s signal, and $\eta_m$ denotes the BackCom reflection coefficient, $0\leq \eta_m\leq 1$ \cite{wongwcnc20}. The power of $s_m$, $0\leq m\leq M$, is assumed to be normalized.  

The signal received at the base station is given by
\begin{align}
y_{\rm BS} = \sum^{M}_{m=1} \sqrt{\eta_m P_0}h_m^2s_0s_m +s_{\rm SI}+n_{\rm BS},
\end{align}
and the signal received by the downlink user is given by
\begin{align}\label{down model}
y_{\rm D} =\sqrt{P_0}h_0s_0 + \sum^{M}_{m=1} \sqrt{\eta_m P_0}g_mh_ms_0s_m +n_{\rm D},
\end{align}
where $s_{\rm SI}$ denotes the FD self-interference, $h_0$ denotes the channel gain between ${\rm U}_0$ and the base station, $g_m$ denotes the channel gain between ${\rm U}_0$ and ${\rm BD}_m$, $n_{\rm BS}$ and $n_{\rm D}$ denote the noises at the base station and ${\rm U}_0$, respectively. As  in \cite{7372448},   $s_{\rm SI}$ is assumed to be complex Gaussian distributed, i.e., $s_{\rm SI}\sim CN(0,\alpha P_0|h_{\rm SI}|^2)$, where $h_{\rm SI}$ denotes the self-interference channel and is  assumed to be complex Gaussian distributed, i.e., $h_{\rm SI}\sim CN(0,1)$, and $0\leq \alpha\ll 1$ indicates the amount of FD residual self-interference. It is assumed that both $n_{\rm BS}$ and $n_{\rm D}$ have the same power,  denoted by $\sigma^2$.

The uplink sum rate achieved by BAC-NOMA  transmission can be written as follows:
\begin{align}\label{sum ratex}
R_{\rm sum} = \log\left( 1+\frac{ \sum^{M}_{m=1} \eta_m P_0|h_m|^4|s_0|^2}{\alpha P_{0} |h_{\rm SI}|^2+\sigma^2}\right),
\end{align}
which is conditioned on $|s_0|^2$. 
 The downlink user's data rate is given by
 \begin{align} \label{down rate}
  R_{\rm D} = \log\left( 1+\frac{P_0|h_0|^2}{
  \sum^{M}_{m=1}  \eta_m P_0 |g_m|^2|h_m|^2   +\sigma^2
  }
 \right).
 \end{align}
{\it Remark 1:} We  note that $s_0$ is perfectly known  by the base station priori to transmission, which means that the expression of $R_{\rm sum}$ shown in \eqref{sum ratex} can be used as an objective function for resource allocation.   However, resource allocation based on $R_{\rm sum}$ requires  strong assumptions, e.g., an individual  resource allocation policy might be needed for each   choice of $|s_0|^2$. Therefore, in the next section, a more computationally efficient way will be used by carrying out resource allocation based on the expectation of the uplink sum rate with respect to $|s_0|^2$, i.e., $\mathcal{E}_{|s_0|^2}\left\{R_{\rm sum}\right\}$, where $\mathcal{E}\{\cdot\}$ denotes an expectation operation. The difference between the resource allocation polices based on $\mathcal{E}_{|s_0|^2}\left\{R_{\rm sum}\right\}$ and $ R_{\rm sum} $ will also be discussed. 

\section{Optimal Resource Allocation}
The resource allocation   problem considered in this paper is to maximize the  sum rate of the uplink BackCom devices.  As discussed in the previous section,  $\mathcal{E}_{|s_0|^2}\left\{R_{\rm sum}\right\}$ is a more appropriate   choice, and it can be evaluated  as follows: 
\begin{align}
\bar{R}_{\rm sum} = \int^{\infty}_{0}\log\left( 1+\frac{ \sum^{M}_{m=1} \eta_m P_0|h_m|^4x}{\alpha P_{0} |h_{\rm SI}|^2+\sigma^2}\right)f_{|s_0|^2}(x)dx ,
\end{align}
where  $f_{|s_0|^2}(x)$ denotes the probability density function (pdf) of  $|s_0|^2$.  
By assuming that  $s_0$ is complex Gaussian distributed with zero mean and unit variance, the pdf of  $|s_0|^2$ is given by $f_{|s_0|^2}(x)=e^{-x}$. With some algebraic manipulation,  $\bar{R}_{\rm sum} $ can be expressed as follows:
\begin{align}
&\bar{R}_{\rm sum}  \\\nonumber &=  -\log_2(e)  e^{\frac{\alpha P_{0} |h_{\rm SI}|^2+\sigma^2}{ \sum^{M}_{m=1} \eta_m P_0|h_m|^4}}E_i\left(-\frac{\alpha P_{0} |h_{\rm SI}|^2+\sigma^2}{ \sum^{M}_{m=1} \eta_m P_0|h_m|^4}\right),
\end{align}
where  $E_i(\cdot)$ denotes the exponential integral function \cite{GRADSHTEYN}. Therefore, the considered throughput maximization problem can be formulated as follows:
  \begin{problem}\label{pb:1} 
  \begin{alignat}{2}
\underset{P_0,\eta_m}{\rm{max}} &\quad    
-  e^{\frac{\alpha P_{0} |h_{\rm SI}|^2+\sigma^2}{ \sum^{M}_{m=1} \eta_m P_0|h_m|^4}}E_i\left(-\frac{\alpha P_{0} |h_{\rm SI}|^2+\sigma^2}{ \sum^{M}_{m=1} \eta_m P_0|h_m|^4}\right) \label{1obj:1} \\
\rm{s.t.} & \quad  
\log\left( 1+\frac{P_0|h_0|^2}{
  \sum^{M}_{m=1}  \eta_m P_0 |g_m|^2|h_m|^2   +\sigma^2
  }
 \right) \geq R_0\label{1st:1}
\\
& \quad 0\leq \eta_m\leq 1, \quad 1\leq m \leq M \label{1st:2}
\\
& \quad 0\leq P_0\leq P_{\rm max} \label{1st:3},
  \end{alignat}
\end{problem} 
where $R_0$ denotes ${\rm U}_0$'s target data rate and $P_{\rm max}$ denotes the transmit power budget at the base station. We note that the use of constraint \eqref{1st:1} is to ensure that  the spectrum cooperation between the downlink user and the    multiple  BackCom devices is carried out without degrading     the downlink user's QoS experience. 

{\it Remark 2:} Problem \eqref{pb:1}  is challenging to solve due to the following three reasons. Firstly, its objective function   and    constraint \eqref{1st:1} contain multiplications of the optimization variables. Secondly, both the objective function and   constraint \eqref{1st:1} contain fractional functions of the optimization variables. Thirdly, the existence of the exponential integral function makes it difficult to analyze the concavity of the objective function.  However, problem \eqref{pb:1} can be recasted to an equivalent linear programming (LP) form, as shown in the following. 

By introducing new variables, $P_m\triangleq \eta_m P_0$,   problem \eqref{pb:1} can be recasted as the following equivalent form:
\begin{problem}\label{pb:2} 
  \begin{alignat}{2}
\underset{P_0,P_m}{\rm{max}} &\quad    
-  e^{\frac{\alpha P_{0} |h_{\rm SI}|^2+\sigma^2}{ \sum^{M}_{m=1} P_m|h_m|^4}}E_i\left(-\frac{\alpha P_{0} |h_{\rm SI}|^2+\sigma^2}{ \sum^{M}_{m=1} P_m|h_m|^4}\right) \label{2obj:1} \\
\rm{s.t.} & \quad  
 P_0|h_0|^2 -
 \epsilon_0 \sum^{M}_{m=1}  P_m |g_m|^2|h_m|^2   -\epsilon_0\sigma^2
    \geq 0\label{2st:1}
\\
& \quad 0\leq P_m\leq P_0 \label{2st:2}
\\
& \quad 0\leq P_0\leq P_{\rm max} \label{2st:3},
  \end{alignat}
\end{problem}  
where $ \epsilon_0=2^{R_0}-1$. As shown in problem \eqref{pb:2},   multiplications between the optimization variables are avoided. 

In order to convert the addressed problem to an LP form, problem \eqref{pb:2}  needs to be first expressed in a more compact way by using the following vector/matrix definitions: 
\begin{align}
\mathbf{a} =& \begin{bmatrix} 0&|h_1|^4 &\cdots &|h_M|^4 \end{bmatrix}^T,\\\nonumber
\mathbf{d} =& \begin{bmatrix} \alpha  |h_{\rm SI}|^2&0&\cdots &0 \end{bmatrix}^T,\\\nonumber
\mathbf{a}_1 =& \begin{bmatrix}  \epsilon_0    |g_1|^2|h_1|^2 &\cdots &\epsilon_0    |g_M|^2|h_M|^2 \end{bmatrix},
\\\nonumber
\mathbf{A} = &\begin{bmatrix}- |h_0|^2 &\mathbf{a}_1\\ \mathbf{0}_{M\times 1} &-\mathbf{I}_M \\ -\mathbf{1}_{M\times 1} &\mathbf{I}_M
\\ -1&\mathbf{0}_{1\times M}\\ 1&\mathbf{0}_{1\times M}
\end{bmatrix},
\\\nonumber
\mathbf{b} =& \begin{bmatrix} -\epsilon_0\sigma^2 &\mathbf{0}_{1\times M}&\mathbf{0}_{1\times M}
&0&P_{\rm max}
 \end{bmatrix}^T.
\end{align}
By using the above definitions, problem \eqref{pb:2}  can be recasted as the following equivalent and more compact form: 
\begin{problem}\label{pb:3} 
  \begin{alignat}{2}
\underset{P_0,P_m}{\rm{max}} &\quad    
-  e^{\frac{\mathbf{d}^T\mathbf{p}+\sigma^2}{ \mathbf{a}^T\mathbf{p}}}E_i\left(-\frac{\mathbf{d}^T\mathbf{p}+\sigma^2}{ \mathbf{a}^T\mathbf{p}}\right) \label{3obj:1} \\
\rm{s.t.} & \quad  
\mathbf{A}\mathbf{p}\leq \mathbf{b}\label{3st:2},
  \end{alignat}
\end{problem}  
where $\mathbf{p}=\begin{bmatrix}P_0&\cdots &P_M \end{bmatrix}^T$. 
To remove the exponential integral function from the considered optimization problem, the following lemma will be needed.

\begin{lemma}\label{lemma1}
$f(x)\triangleq  -e^{\frac{1}{x}}E_i\left(- \frac{1}{x}\right)$, $x\geq0$, is a monotonically increasing function of $x$. 
\end{lemma}
\begin{proof} The lemma can be proved by  showing that the first order derivative of $f(x)$ is non-negative, which requires the derivative of the exponential integral function. We first note that an integral expression of the exponential integral function  is given by \cite{GRADSHTEYN}
\begin{align}
E_i(x) = \int^{x}_{\infty}\frac{e^t}{t}dt, \quad x \leq 0,
\end{align}
which means that the first-order derivative of $E_i(x)$ is given by
\begin{align}\label{eix}
\frac{dE_i(x)}{dx} = \frac{e^x}{x}.
\end{align}
By using \eqref{eix}, we can also obtain the following:   $\frac{dE_i\left(-\frac{1}{x}\right)}{dx} = -\frac{e^{-\frac{1}{x}}}{x}$. 
Therefore, the first order derivative of $f(x)$ can be obtained as follows:
\begin{align}\nonumber 
\frac{df(x)}{dx} =& \frac{e^{\frac{1}{x}}}{x^2}E_i\left(- \frac{1}{x}\right) +e^{\frac{1}{x}} \frac{e^{-\frac{1}{x}}}{x}
\\\label{first orderdd}
=& \frac{e^{\frac{1}{x}}}{x^2}E_i\left(- \frac{1}{x}\right) +\frac{1}{x}.
\end{align}
We note that $E_i\left(- \frac{1}{x}\right)\leq 0$, for $x\geq 0$, which means that \eqref{first orderdd} is not necessarily non-negative. To show that $\frac{df(x)}{dx} $ is non-negative,  we define the following function: 
\begin{align}
g(x) = E_i\left(- \frac{1}{x}\right) +xe^{-\frac{1}{x}}.
\end{align}
It is straightforward to show that $g(x)\geq 0$ is equivalent to $\frac{df(x)}{dx} \geq 0$. The first order derivative of $g(x)$ can be obtained as follows:
\begin{align}
\frac{dg(x)}{dx}  =& -\frac{e^{-\frac{1}{x}}}{x}+e^{-\frac{1}{x}}+xe^{-\frac{1}{x}}\frac{1}{x^2}= e^{-\frac{1}{x}}\geq 0,
\end{align}
which  means that $g(x)$ is a monotonically increasing function of $x$. Therefore, $g(x)$ can be lower bounded as follows:
\begin{align}\label{gxx}
g(x)\geq g(0)= E_i\left(- \frac{1}{x}\right) +xe^{-\frac{1}{x}}=0,
\end{align} 
where the last step follows from the following approximation  that $E_i\left(x\right) \approx \frac{e^x}{x}$, for $x\rightarrow -\infty$ \cite{GRADSHTEYN}. \eqref{gxx} implies that 
$\frac{df(x)}{dx} \geq 0$. Or in other words, $f(x)$ is a monotonically increasing function of $x$, and the lemma is proved.  
\end{proof}

By using Lemma \ref{lemma1}, the optimal solution of problem \eqref{pb:3} can be found by recasting problem \eqref{pb:3}  as follows: 
\begin{problem}\label{pb:4} 
  \begin{alignat}{2}
\underset{P_0,P_m}{\rm{max}} &\quad    
\frac{ \mathbf{a}^T\mathbf{p}}{\mathbf{d}^T\mathbf{p}+\sigma^2} \label{4obj:1} \\
\rm{s.t.} & \quad  
\mathbf{A}\mathbf{p}\leq \mathbf{b}\label{4st:2}.
  \end{alignat}
\end{problem}  
Note that problem \eqref{pb:4} is a linear fractional problem. As a special case of quasi-concave problems, the bi-section method can be used to solve this maximization  problem, by converting it to a series of  convex feasibility problems. In contrast, a more computationally efficient way is to convert problem \eqref{pb:4} to an LP form as follows: \cite{Boyd}
 \begin{problem}\label{pb:5} 
  \begin{alignat}{2}
\underset{\mathbf{y},z}{\rm{max}} &\quad    
 \mathbf{a}^T\mathbf{y}\label{5obj:1} \\
\rm{s.t.} & \quad  
 \mathbf{A}\mathbf{y}\leq \mathbf{b}z\label{5st:1} 
, \quad \mathbf{d}^T\mathbf{y}+\sigma^2z=1,  
 \quad z\geq 0 ,
  \end{alignat}
\end{problem} 
where $\mathbf{y}= \frac{ \mathbf{p}}{\mathbf{d}^T\mathbf{p}+\sigma^2}$ and $z= \frac{1}{\mathbf{d}^T\mathbf{p}+\sigma^2}$. For the simple LP form shown in \eqref{pb:5}, various optimization solvers, such as CVX \cite{Boyd}, can be used to find its optimal solution.  By using the obtained  solution of problem \eqref{pb:5}, denoted by $\mathbf{y}^*$ and $z^*$,  the optimal value of problem \eqref{pb:1}, denoted by $p^*$, can be found as follows. First, we note the following relationship between $\mathbf{y}^*$  and the optimal solution of problem \eqref{pb:4}, denoted by $\mathbf{p}^*$, as follows:  
\[
\frac{ \mathbf{a}^T\mathbf{p}^*}{\mathbf{d}^T\mathbf{p}^*+\sigma^2} =\mathbf{a}^T\mathbf{y}^*. 
\]
Since the objective function  of problem \eqref{pb:1} is a function of $\frac{ \mathbf{a}^T\mathbf{p}}{\mathbf{d}^T\mathbf{p}+\sigma^2} $, $p^*$ can be expressed as a function of $\mathbf{y}^*$ as follows: 
\begin{align}
p^*=&-  \log_2(e)e^{\frac{\mathbf{d}^T\mathbf{p}+\sigma^2}{ \mathbf{a}^T\mathbf{p}^*}}E_i\left(-\frac{\mathbf{d}^T\mathbf{p}^*+\sigma^2}{ \mathbf{a}^T\mathbf{p}^*}\right) \\\nonumber =&-  \log_2(e)e^{\frac{1}{\mathbf{a}^T\mathbf{y}^*}}E_i\left(-\frac{1}{\mathbf{a}^T\mathbf{y}^*}\right).  
\end{align}

\subsection{An Alternative Problem Formulation }
The aforementioned resource allocation   problem is to maximize the average sum rate, $\mathcal{E}_{|s_0|^2}\left\{R_{\rm sum}\right\}$.  Alternatively, by using the instantaneous sum rate, $R_{\rm sum}$, as the objective function, the throughput maximization problem can be formulated as follows:  
  \begin{problem}\label{pb:6} 
  \begin{alignat}{2}
\underset{P_0,\eta_m}{\rm{max}} &\quad    
\log\left( 1+\frac{ \sum^{M}_{m=1} \eta_m P_0|h_m|^4|s_0|^2}{\alpha P_{0} |h_{\rm SI}|^2+\sigma^2}\right) \label{6obj:1} \\
\rm{s.t.} & \quad  
\eqref{1st:1},   \eqref{1st:2}, \eqref{1st:2}.
  \end{alignat}
\end{problem} 
By using the fact that $\log(1+x)$ is a monotonically increasing function of $x$ and following the steps to reformulate  problem \eqref{pb:1}, problem  \eqref{pb:6}  can also be recasted to an equivalent LP form as follows: 
\begin{problem}\label{pb:7} 
  \begin{alignat}{2}
\underset{\mathbf{y},z}{\rm{max}} &\quad    
 \bar{\mathbf{a}}^T\mathbf{y}\label{7obj:1}, \quad
\rm{s.t.} & \quad   \eqref{5st:1} ,
  \end{alignat}
\end{problem} 
where   $\bar{\mathbf{a}}$ is defined as follows: 
\begin{align}
\bar{\mathbf{a}} =& \begin{bmatrix} 0&|h_1|^4|s_0|^2&\cdots &|h_M|^4|s_0|^2 \end{bmatrix}^T.
\end{align} 
Note that   $\bar{\mathbf{a}} = |s_0|^2{\mathbf{a}} $, which means that the objective function of problem \eqref{pb:7} is a scaled version of the objective of problem \eqref{pb:5}. Because $|s_0|^2\geq  0$, problem \eqref{pb:7} is equivalent  to problem \eqref{pb:5}, which means that  problems \eqref{pb:1} and \eqref{pb:7} are equivalent. Therefore,  the two different formulations shown in \eqref{pb:1} and \eqref{pb:7}  yield the same resource allocation polices.

 \begin{figure}[t]\centering \vspace{-0em}
    \epsfig{file=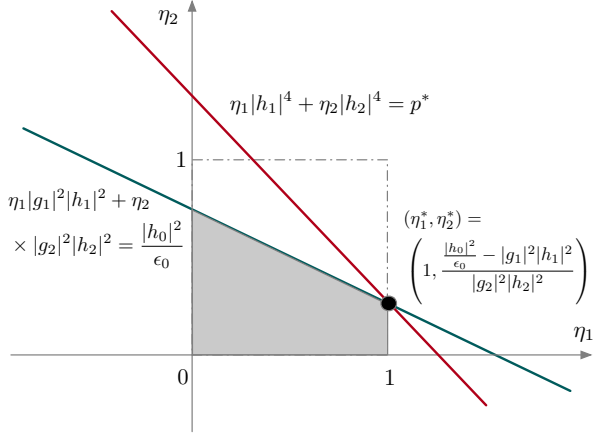, width=0.43\textwidth, clip=}\vspace{-0.5em}
\caption{ Illustration for the feasibility region of the approximated  LP problem shown in \eqref{pb:8}, and the optimal solution for $\eta_m$ shown in \eqref{optimal eta}.  It is assumed that  $ |g_1|^2|h_1|^2    +  |g_2|^2|h_2|^2 
    > \frac{ |h_0|^2}{ \epsilon_0 } $, $|h_1|^2=|h_2|^2$ and $|g_1|^2<|g_2|^2$, i.e.,    ${\rm BD}_2$ has a stronger connection  to ${\rm U}_0$ than ${\rm BD}_1$.     }\label{fig2}\vspace{-1.5em}
\end{figure}

\subsection{Closed-Form Solutions for  the Two-User   Special Cases} \label{subsection special case}
For the two-user case with high signal-to-noise ratio (SNR), i.e., $M=2$ and $\sigma_n^2 \rightarrow 0$, the formulated problem can be approximated  as follows:
  \begin{problem}\label{pb:8} 
  \begin{alignat}{2}
\underset{\eta_m}{\rm{max}} &\quad  \eta_1|h_1|^4+\eta_2 |h_2|^4  \label{8obj:1} \\
\rm{s.t.} & \quad  
  \eta_1 |g_1|^2|h_1|^2    + \eta_2 |g_2|^2|h_2|^2 
    \leq \frac{ |h_0|^2}{ \epsilon_0 } \label{8st:1}
\\
& \quad 0\leq \eta_1\leq 1, \quad 0\leq \eta_2\leq 1 \label{8st:2}
,
  \end{alignat}
\end{problem} 
which is not related to $P_0$  due to the used high-SNR assumption.  Due to space limitations, we only focus on the case that the two users have the same channel gains to the base station, and ${\rm BD}_2$ has a stronger connection to ${\rm U}_0$ than  ${\rm BD}_1$, i.e., $|h_1|^2=|h_2|^2$ and $|g_1|^2<|g_2|^2$. The optimal solutions for other cases can be obtained in  a straightforward manner.

Note that constraint \eqref{8st:2} defines a square-shape region, and constraint \eqref{8st:1} defines a half-space. Provided that  $ |g_1|^2|h_1|^2    +  |g_2|^2|h_2|^2 
    \leq \frac{ |h_0|^2}{ \epsilon_0 } $, the half-space fully covers  the square-shape region. Therefore, the square-shape region is the feasibility region of the LP,  which means that $\eta_1=\eta_2=1$  \cite{Boyd}.   This is expected because   $ |g_1|^2|h_1|^2    +  |g_2|^2|h_2|^2 
    \leq \frac{ |h_0|^2}{ \epsilon_0 } $ means that the channels of the BackCom devices are weak, and hence    both the  devices can use  their maximal transmit power,   $\eta_1=\eta_2=1$, while still guaranteeing ${\rm U}_0$'s QoS requirement.

Provided that  $ |g_1|^2|h_1|^2    +  |g_2|^2|h_2|^2 
    > \frac{ |h_0|^2}{ \epsilon_0 } $, the intersection between the half-space and the square-shape region is the feasibility region of the LP. For the considered case,   i.e., $|h_1|^2=|h_2|^2$ and $|g_1|^2<|g_2|^2$, the feasibility region of the LP is shown  in Fig. \ref{fig2}. With some straightforward algebraic manipulations, the closed-form expressions for the optimal solutions of problem \eqref{pb:8} can be obtained as follows: \cite{Boyd}
    \begin{align}\label{optimal eta}
    \eta_1^*=1,\quad \& \quad \eta_2^*=   \frac{ \frac{ |h_0|^2}{ \epsilon_0 }-|g_1|^2|h_1|^2 }{  |g_2|^2|h_2|^2 }.
    \end{align}
The solutions in  \eqref{optimal eta} are also expected, since the device with a weak connection to ${\rm U}_0$ can still use its full transmit power, i.e., $\eta_1^*=1$, but the transmit power of the device with a strong connection  to  ${\rm U}_0$  needs to be carefully controlled in order to avoid too much performance degradation to ${\rm U}_0$.

 \begin{figure}[t]\centering \vspace{-1em}
    \epsfig{file=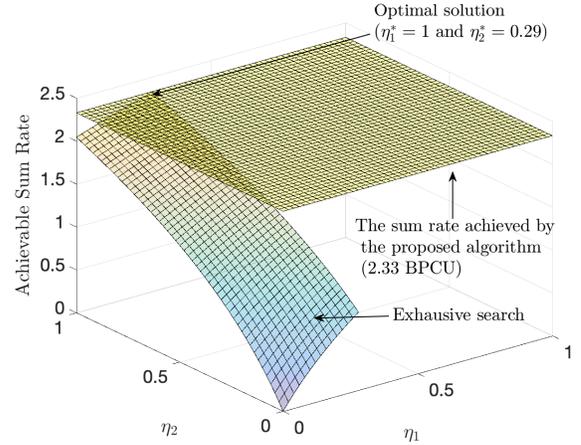, width=0.45\textwidth, clip=}\vspace{-0.5em}
\caption{ A deterministic study for the optimality of the obtained solution, where random fading is omitted and only path loss effects are considered. $M=2$,  $\alpha=0.005$. $|s_0|^2=1$. $P_0^*=0.043$,     $R_0=1$ bit per channel use (BPCU),    and the locations of   ${\rm BD}_1$ and ${\rm BD}_2$ are fixed at   $(-2,0)$ m and $(2,0)$ m, respectively.         }\label{fig3}\vspace{-1.5em}
\end{figure}

\vspace{-1em}

\section{Numerical Studies}\label{section simulation}
In this section, the performance of the proposed BAC-NOMA scheme is studied by using computer simulations, where   the path loss exponent is set as $3$,   $P_{\rm max}=20$ dBm,   $\sigma_n^2=-94$ dBm, the base station and the downlink user are located at $(0,0)$ m and $(3,0)$ m, respectively.

In Fig. \ref{fig3}, a deterministic study  is carried out to evaluate the optimality of  the obtained solution,  where random fading is omitted  and only path loss effects are considered. In particular, by using CVX to solve problem \eqref{pb:8}, the optimal solutions, $P_0^*$, $\eta_1^*$, and $\eta_2^*$, can be obtained. Exhaustive search  is also carried out to search the optimal choices for $\eta_1$ and $\eta_2$, by fixing   $P_0=P_0^*$. As shown in Fig. \ref{fig3}, the  optimal solutions obtained by the proposed scheme perfectly match  those obtained by exhaustive search. Furthermore,  we note that   the setups for the deterministic study satisfy those requirements outlined in Section \ref{subsection special case}, and the fact that the optimal solutions observed  in Fig. \ref{fig3} perfectly fit the closed-form expressions in  \eqref{optimal eta}   verifies the accuracy of   the analysis     in Section \ref{subsection special case}.    

In Figs. \ref{fig4} and \ref{fig5}, a general multi-user scenario   is considered. In particular, the BackCom devices are uniformly deployed in a square of edge length $5$ m and with the base station located at its center.  Both large-scale path loss and   small-scale multi-path fading are considered. The following two benchmark schemes are used. The first one  is termed BAC-OMA, where the $M$  BackCom devices are scheduled to be paired with the downlink user in a round-robin manner.  Optimal resource allocation is also considered for BAC-OMA, which is a special case of  problem \eqref{pb:1} by choosing  $M=1$. The second benchmark scheme is BAC-NOMA with random resource allocation. We note that it is not  trivial to find  random choices of $\{P_0,\eta_1, \cdots, \eta_M\}$ because they   need to satisfy constraints \eqref{1st:1}, \eqref{1st:2}, and \eqref{1st:3}.  For the conducted simulations, the random choices of $\{P_0,\eta_1, \cdots, \eta_M\}$ are found by solving the following feasibility problem: 
  \begin{problem}\label{pb:9} 
  \begin{alignat}{2}
 {\rm find}&\quad    
 P_0,\eta_1, \cdots, \eta_M \label{9obj:1} \quad\quad
\rm{s.t.} & \quad  \eqref{1st:1}, \eqref{1st:2}, \eqref{1st:3},
  \end{alignat}
\end{problem} 
which is an LP problem and hence can be solved similarly to problem \eqref{pb:8}.

In Fig. \ref{fig4}, the performance of the three transmission schemes is shown as a function of the number of BackCom devices, $M$.  As can be seen from the figure, the two BAC-NOMA schemes can realize significant performance gains over BAC-OMA, which is due to the reason that multiple BackCom devices can be served simultaneously by the BAC-NOMA schemes. As a result, the more devices there are, the larger sum rates the BAC-NOMA schemes can achieve. Between the two BAC-NOMA schemes, Fig. \ref{fig4} demonstrates that the use of the proposed resource allocation algorithm can realize a larger sum rate, and this performance gain can be further  increased by increasing $M$. It is worth to point out that the two BAC-NOMA schemes require similar computational complexity, since the random scheme still needs to solve an LP  problem. Another interesting observation from Fig. \ref{fig4} is that the performance of BAC-OMA is not a function of $M$, since for BAC-OMA, a single uplink BackCom device is randomly scheduled, regardless how many devices there are. 
 
  \begin{figure}[t]\centering \vspace{-1em}
    \epsfig{file=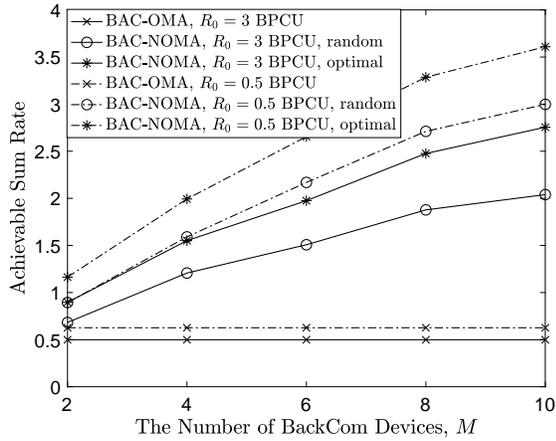, width=0.4\textwidth, clip=}\vspace{-0.5em}
\caption{ Average sum rates achieved by the three transmission schemes as a function of the number of BackCom devices $M$.  $\alpha=0.01$     }\label{fig4}\vspace{-1.5em}
\end{figure}

In Fig. \ref{fig5}, the performance of the three considered transmission schemes is shown as a function of the FD self-interference coefficient, $\alpha$, where different choices of $M$ are used. Note that a single curve is shown for BAC-OMA, since its performance is not affected by the choice of $M$.  As can be observed from the figure, the BAC-NOMA scheme with optimal resource allocation yields the best performance among the three schemes, which is consistent to Fig. \ref{fig4}. Recall that $\alpha$ implies how much residual self-interference exists. Therefore, it is expected that the performance of the three FD transmission schemes should be degraded by increasing $\alpha$, which is indeed observed in Fig. \ref{fig5}. Another important  observation from the figure is that the performance gain of NOMA over OMA is diminishing by increasing $\alpha$, which can be explained as follows. When $\alpha$ is very large, \eqref{sum ratex} indicates that   FD self-interference becomes severe, i.e., this is a low SNR scenario, where the use of NOMA does not bring a significant performance gain, as demonstrated by the existing studies for  NOMA \cite{Nomading,7676258}.

 \begin{figure}[t]\centering \vspace{-1em}
    \epsfig{file=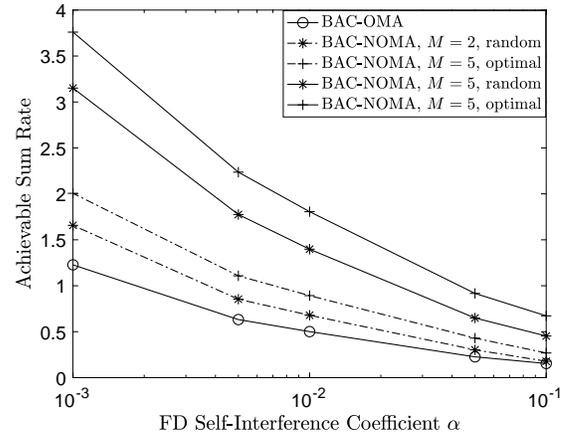, width=0.4\textwidth, clip=}\vspace{-0.5em}
\caption{ Average sum rates achieved by the three transmission schemes as a function of the FD residual interference coefficient $\alpha$.  $R_0=3$ BPCU.     }\label{fig5}\vspace{-1.5em}
\end{figure}
\vspace{-1em}
\section{Conclusions}
In this letter,    BAC-NOMA has been proposed to  realize  effective  spectrum and energy    cooperation among uplink and downlink transmission.  The   problem of maximize uplink throughput while  suppressing  the interference between downlink and uplink transmission has been formulated and solved. Computer simulations have also been provided to  demonstrate the superior  performance   of   BAC-NOMA. 
 
\vspace{-1em}
 \bibliographystyle{IEEEtran}
\bibliography{IEEEfull,trasfer}

  \end{document}